\documentclass[a4paper,UKenglish,cleveref, autoref]{lipics-v2019}

\hideLIPIcs
\usepackage{amsthm}
\usepackage{amsmath}
\usepackage{algorithm2e}
\usepackage{tikz}
\usepackage{subcaption}
\newtheorem{problem}{Problem}
\newcommand{\ltdots}{..}

\title{Chaining with Overlaps Revisited} 

\titlerunning{Chaining with Overlaps}

\author{Veli M\"akinen}{Department of Computer Science, University of Helsinki, Finland }{veli.makinen@helsinki.fi}{https://orcid.org/0000-0003-4454-1493}{}
\author{Kristoffer Sahlin}{Department of Mathematics, Science for Life Laboratory, Stockholm University, Sweden}{ksahlin@math.su.se}{https://orcid.org/0000-0001-7378-2320}{}

\authorrunning{V. M\"akinen, K. Sahlin}

\Copyright{Veli M\"akinen and Kristoffer Sahlin}

\ccsdesc[500]{Theory of computation~Pattern matching}
\ccsdesc[500]{Theory of computation~Dynamic programming}
\ccsdesc[500]{Applied computing~Genomics}

\keywords{Sparse Dynamic Programming, Chaining, Maximal Exact Matches, Longest Common Subsequence}

\category{}

\relatedversion{Final version to appear in CPM 2020.}

\supplement{}

\nolinenumbers 

\begin{document}

\maketitle

\begin{abstract}
Chaining algorithms aim to form a semi-global alignment of two sequences based on a set of anchoring local alignments as input. Depending on the optimization criteria and the exact definition of a chain, there are several $O(n \log n)$ time algorithms to solve this problem optimally, where $n$ is the number of input anchors. 

In this paper, we focus on a formulation allowing the anchors to overlap in a chain. This formulation was studied by Shibuya and Kurochkin (WABI 2003), but their algorithm comes with no proof of correctness. We revisit and modify their algorithm to consider a strict definition of precedence relation on anchors, adding the required derivation to convince on the correctness of the resulting algorithm that runs in $O(n \log^2 n)$ time on anchors formed by exact matches. With the more relaxed definition of precedence relation considered by Shibuya and Kurochkin or when anchors are non-nested such as matches of uniform length ($k$-mers), the algorithm takes $O(n \log n)$ time. 

We also establish a connection between chaining with overlaps and the widely studied \emph{longest common subsequence} problem.
\end{abstract}

\section{Introduction}

As optimal alignment of two strings takes quadratic time (which has recently been shown to be conditionally hard to improve \cite{BI15}), there have been several attempts to avoid this bottleneck. One such technique is \emph{sparse dynamic programming} \cite{Epp92}, where a sparse set of cells of the dynamic programming matrix is identified whose computation is sufficient in computing the optimal alignment. This does not avoid the quadratic dependency in the worst case, so a slightly more heuristic \emph{chaining} approach has been introduced in the context of computational genomics: Given a precomputed set of plausible \emph{anchoring local matches}, extract a chain of matches that forms a good (semi-global) alignment. 

In this paper, we investigate a chaining formulation that takes properly the overlaps between anchors into account. Namely, if anchors are not allowed to overlap in the solution, there are already several $O(n \log n)$ time solutions for various formulations of the chaining problem \cite{MM95,FMW97,AO03,AO05}, where $n$ is the number of anchors. Some of the solutions and extensions focus on asymmetric measures, where overlaps are allowed in one of the strings \cite{MSY12,Maketal19}, or add other features that make the problem even harder \cite{UMR11}. While these formulations are useful in different contexts, this is an undesirable consequence in, \emph{e.g.}, string alignment, where the solution may be different depending on which string is used to traverse the ordered anchors, and specifically the solution may overcount the amount of aligned characters.

The fully symmetric chaining variant allows arbitrary overlaps, guarantees not to overcount the amount of aligned characters, and in addition, is particularly important for its connections to the \emph{Longest Common Subsequence} problem (LCS): An optimal chain in this formulation corresponds to a LCS of the input strings, restricted to the matches included in the anchors. As far as we know, except for trivial $O(n^2)$ time solutions, only Shibuya and Kurochkin \cite{SK03} have proposed a solution aiming to solve the fully symmetric case of allowing overlaps of anchors in both strings simultaneously. 

We revisit the algorithm by Shibuya and Kurochkin \cite{SK03} and propose a modification that takes into account a strict order for the anchors. This modified algorithm runs in $O(n \log^2 n)$ time on exact matches as input. When relaxing the precedence order or when the input consist of non-nested anchors such as $k$-mer matches, the algorithm can be simplified to take $O(n \log n)$ time. The resulting algorithms are slightly simpler than the original \cite{SK03}, requiring only a general data structure for semi-dynamic range maximum queries, while  the original uses in addition a tailored structure. We also provide detailed derivation of the algorithms, while the original \cite{SK03} comes with no proof of correctness. Finally, we show that the relaxed chaining problem also solves a restricted version of the LCS problem.

\section{Chaining problems}

Let $T$ be a long text string and $P$ short pattern string. An \emph{anchor} interval pair $([a \ltdots b],[c \ltdots d])$ denotes a \emph{match} between $T[a\ltdots b]$ and $P[c\ltdots d]$. For now, we assume these matches are precomputed, and they could be either full identities or close similarities. We often abstract out the original source of the anchors referring $[a \ltdots b]$ as an interval in the \emph{first dimension} and $[c \ltdots d]$ as an interval in the \emph{second dimension}. We denote the endpoints of the intervals in anchor $I$ as $I.x$ for $x \in \{a,b,c,d\}$.
We assume the endpoints to be positive integers. 

Given two anchors $I'$ and $I$ we define two relations: \emph{precedence} and \emph{overlap}. The former is denoted $I'\prec I$ and this relation holds whenever $I'.a< I.a$, $I'.b< I.b$, $I'.c< I.c$, and $I'.d< I.d$. The latter is denoted $I'\cap I$ and holds whenever $[I'.a \ltdots I'.b] \cap [I.a \ltdots I.b]\neq \emptyset$ or $[I'.c \ltdots I'.d] \cap [I.c \ltdots I.d]\neq \emptyset$. The complement of the overlap relation is denoted $\neg I'\cap I$ (an empty intersection is interpreted as truth value False). We use the overlap relation only for $I' \prec I$. Figure~\ref{fig:precedence} illustrates these concepts.

\begin{figure}
    \centering
    \begin{subfigure}[t]{0.5\textwidth}
        \centering
        \begin{tikzpicture}
        \draw (0,2) to (5,2) node {\verb+  +$T$};
        \draw (0,0) to (5,0) node {\verb+  +$P$};
        \draw[densely dotted,thick] (1.5,2) node[rectangle,draw](Iptop) {\verb+      +};
        \draw (3.5,2) node[rectangle,draw](Itop) {\verb+    +};
        \draw[densely dotted,thick] (2,0) node[rectangle,draw](Ipbottom) {\verb+      +};
        \draw (3.5,0) node[rectangle,draw](Ibottom) {\verb+    +};
        \draw[-] (Iptop) to (Ipbottom);
        \draw[-] (Itop) to (Ibottom);
        \end{tikzpicture}
        \caption{No overlaps}
    \end{subfigure}%
    ~
    \begin{subfigure}[t]{0.5\textwidth}
        \centering
        \begin{tikzpicture}
        \draw (0,2) to (5,2) node {\verb+  +$T$};
        \draw (0,0) to (5,0) node {\verb+  +$P$};
        \draw[densely dotted,thick] (1.5,2) node[rectangle,draw](Iptop) {\verb+      +};
        \draw (3.5,2) node[rectangle,draw](Itop) {\verb+    +};
        \draw[densely dotted,thick] (3,0) node[rectangle,draw](Ipbottom) {\verb+      +};
        \draw (3.5,0) node[rectangle,draw](Ibottom) {\verb+    +};
        \draw[-] (Iptop) to (Ipbottom);
        \draw[-] (Itop) to (Ibottom);
        \end{tikzpicture}
        \caption{One-sided overlaps}
    \end{subfigure}
    \\
    \vspace{1cm}
    \begin{subfigure}[t]{0.5\textwidth}
        \centering
        \begin{tikzpicture}
        \draw (0,2) to (5,2) node {\verb+  +$T$};
        \draw (0,0) to (5,0) node {\verb+  +$P$};
        \draw[densely dotted,thick] (3.1,2) node[rectangle,draw](Iptop) {\verb+      +};
        \draw (3.5,2) node[rectangle,draw](Itop) {\verb+    +};
        \draw[densely dotted,thick] (2.7,0) node[rectangle,draw](Ipbottom) {\verb+      +};
        \draw (3.5,0) node[rectangle,draw](Ibottom) {\verb+    +};
        \draw[-] (Iptop) to (Ipbottom);
        \draw[-] (Itop) to (Ibottom);
        \end{tikzpicture}        
        \caption{Larger overlap in the first dimension}
    \end{subfigure}%
    ~ 
    \begin{subfigure}[t]{0.5\textwidth}
        \centering
        \begin{tikzpicture}
        \draw (0,2) to (5,2) node {\verb+  +$T$};
        \draw (0,0) to (5,0) node {\verb+  +$P$};
        \draw[densely dotted,thick] (2.7,2) node[rectangle,draw](Iptop) {\verb+      +};
        \draw (3.5,2) node[rectangle,draw](Itop) {\verb+    +};
        \draw[densely dotted,thick] (3,0) node[rectangle,draw](Ipbottom) {\verb+      +};
        \draw (3.5,0) node[rectangle,draw](Ibottom) {\verb+    +};
        \draw[-] (Iptop) to (Ipbottom);
        \draw[-] (Itop) to (Ibottom);
        \end{tikzpicture}        
        \caption{Larger overlap in the second dimension}
    \end{subfigure}
    
    \caption{Different scenarios illustrating precedence and overlap of anchors. Dotted and solid rectangles denote anchors $I'$ and $I$, respectively. In all these cases it holds $I' \prec I$. The separation into different cases based on the overlaps is determined by the properties of the chaining algorithms we study in the sequel. In (a) and (b) no overlaps are allowed in the first dimension, while in (c) and (d) anchors are assumed to overlap in the first dimension. \label{fig:precedence}}
\end{figure}
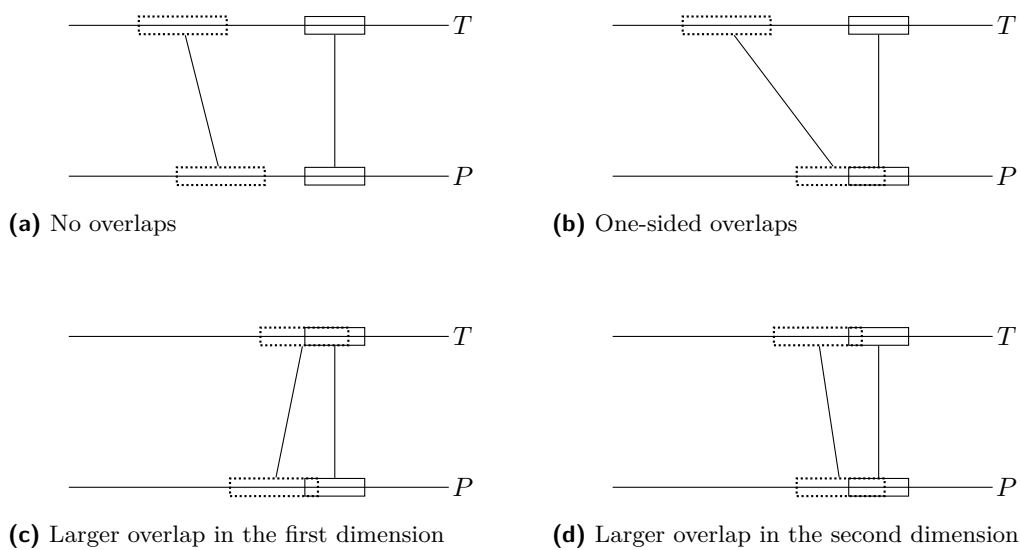

\begin{problem}[Chaining with overlaps\label{prob:colinearchaining}]
Let $A[1..N]$ be an array of anchor interval pairs $([a \ltdots b],[c \ltdots d])$. For each $i$, $1\leq i \leq N$, compute the \emph{symmetric ordered coverage} score \\ $\max_{\text{chains } S^i} \mathtt{coverage}(S^i)$, where
\begin{itemize}
\item $S^i[1..n]$ is an ordered subset (\emph{chain}) of pairs from $A$,
\item $S^i[j-1]\prec S^i[j]$, for all $1< j \leq n$,
\item $S^i[n]=A[i]$, and 
\begin{align*}
\mathtt{coverage}(S^i)=
\Bigg(\sum_{j=1}^{n-1} \min\bigg(&\min(S^i[j+1].a,S^i[j].b+1)-S^i[j].a,
\\&\min(S^i[j+1].c,S^i[j].d+1)-S^i[j].c\bigg)\Bigg)+\\
\min\bigg(&S^i[n].b-S^i[n].a+1,S^i[n].d-S^i[n].c+1\bigg).
\end{align*}
\end{itemize}
\end{problem}

Notice that for chains containing no overlaps, that is, $S^i[j].b<S^i[j+1].a$ and $S^i[j].d<S^i[j+1].c$, 
the measure $\mathtt{coverage}(S^i)$ is just the sum of lengths of the anchors in it, where length is defined as the minimum of the interval lengths. For overlapping cases, only the segment before the overlap is added to the score. For example, let a chain $S$ be $([1..5],[2..6]),([3..8],[5..10])$. Then $\mathtt{coverage}(S)=\min(8-3+1,10-5+1)+\min(\min(3,5+1)-1,\min(5,6+1)-2)=6+2=8$. That is, while the total length of the anchors in $S$ is $10$, their union covers only 8 units in the first dimension. The measure is clearly symmetric, but the term \emph{ordered coverage} requires more insight: Notice that it is not sufficient to measure the size of the union of anchors in a chain independently and take their minimum. Instead, the proposed measure adds to the score, one anchor at the time, the minimum size of the newly covered region. Interpreted through the original source of anchors from string $T$ and $P$, an optimal chain $S$ under this measure induces an alignment between $T$ and $P$ with exactly $\mathtt{coverage}(S)$ matching characters. Asymmetric formulations studied earlier can overestimate this amount. This proposed symmetric ordered coverage measure is thus important especially in various computational genomics applications, where optimal alignments are too expensive to be computed. We establish this alignment connection through the widely studied longest common subsequence problem: see Sect.~\ref{sect:lcs}.   

We develop an $O(N \log^2 N)$ time algorithm to solve this chaining with overlaps problem assuming one additional property of the input:
\begin{itemize}
    \item \emph{Equal Match Length property}: For each anchor $I$ it holds $I.b-I.a=I.d-I.c$.
\end{itemize}
If the set of anchors is computed e.g. by \emph{Maximal Exact Matches (MEMs)} \cite{Gus97}, the input automatically satisfies the Equal Match Length property. 

Our algorithm is based on techniques by Shibuya and Kurochkin \cite{SK03}, who solved a version of the problem with the definition of precedence relaxed to consider only start points of intervals: $I'$ \emph{weakly precedes} $I$ if $I'.a<I.a$ and $I'.c< I.c$. Let us denote this relation $I' \prec^w I$. 

\begin{problem}[Chaining with overlaps and weak precedence\label{prob:colinearchainingweak}]
Let $A[1..N]$ be an array of anchor interval pairs $([a \ltdots b],[c \ltdots d])$. For each $i$, $1\leq i \leq N$, compute the \emph{symmetric weakly ordered coverage} score $\max_{\text{weak chains } S^i} \mathtt{coverage}(S^i)$, defined as in Problem~\ref{prob:colinearchaining}, with the precedence condition relaxed to  
\begin{itemize}
\item $S^i[j-1]\prec^w S^i[j]$, for all $1< j \leq n$.
\end{itemize}
\end{problem}

To see the connection of the problems, consider a chain $S$ for which $S[j-1]\prec^w S[j]$ holds but not $S[j-1]\prec S[j]$ for some $j$. That is, at least one of the intervals of $S[j]$ is nested inside (i.e. is subset of) the corresponding interval of $S[j-1]$. Say $[S[j].a\ltdots S[j].b]$ is nested in $[S[j-1].a\ltdots S[j-1].b]$ with $S[j-1].b-S[j].b\geq     S[j-1].d-S[j].d$ (the other case is symmetric). Consider modifying $S[j-1]$ into $S[j-1]'$, where $S[j-1]'.a=S[j-1].a,S[j-1]'.b=S[j].b-1,S[j-1]'.c=S[j-1].c,$ and $S[j-1]'.d=S[j-1].d-(S[j-1].b-S[j].b)-1$. Assuming Equal Match Length property such adjustment is possible and causes $S[j-1]'\prec S[j]$ without affecting the score. One can thus adjust any chain $S$ for which the weak precedence relation holds into another chain $S'$, where the (strict) precedence relation holds, so that  $\mathtt{coverage}(S)=\mathtt{coverage}(S')$.

As can be seen from the above construction, the two problems are identical when the input anchors are non-nested. This happens e.g. when anchors are matches of uniform length ($k$-mer matches). Even more importantly, if one is only interested in the overall maximum scoring chain, the two problems produce the same result.

\begin{lemma}
\label{lemma:weakisenough}
Assuming Equal Match Length property, the maximum of the solutions from Problem~\ref{prob:colinearchaining} and Problem~\ref{prob:colinearchainingweak} are the same, that is,
\[
\max_{1\leq i \leq N} \max_{\text{chains } S^i} \mathtt{coverage}(S^i)=
\max_{1\leq i \leq N} \max_{\text{weak chains } S^i} \mathtt{coverage}(S^i).
\]
\end{lemma}
\begin{proof}
Consider an optimal chain $S^i[1..n]$ for Problem~\ref{prob:colinearchainingweak}. If $S^i[n-1]\prec S^i[n]$ does not hold, then one of the intervals of $S^i[n]$ is nested inside the corresponding interval of $S^i[n-1]$. This means that for $S^i[n-1]=A[i']$ it holds $\max_{\text{weak chains } S^{i'}} \mathtt{coverage}(S^{i'})\geq \max_{\text{weak chains } S^{i}} \mathtt{coverage}(S^{i})$. Continuing this induction, one observes that there is an overall maximum scoring chain, say $S$, that ends with strict precedence and moreover for all anchors $I$ in this chain holds $I\prec S[n]$. 

Consider now the construction given before this lemma that converts $S$ into $S'$. Adjust the construction so that instead of modifying $S[j-1]$ into $S[j-1]'$, just remove $S[j]$. Repeat this from right to left until strict precedence holds in the whole modified chain $S'$. Such $S'$ is an optimal solution to both problems, as its score remains unchanged during the process. 

To see this, consider the same case as earlier with $[S[j].a\ltdots S[j].b]$ being nested in $[S[j-1].a\ltdots S[j-1].b]$ with $S[j-1].b-S[j].b\geq S[j-1].d-S[j].d$. By induction from the base case of $S[n-1] \prec S[n]$, we know that $S[j]\prec S[j+1]$ and $S[j-1]\prec S[j+1]$ if $S[j-1] \prec^w S[j]$ is the first nested case from the right. 

To see that dropping $S[j]$ is safe, we need to consider cases a) $\neg S[j] \cap S[j+1]$ and $\neg S[j-1] \cap S[j+1]$ b) $S[j] \cap S[j+1]$ and $\neg S[j-1] \cap S[j+1]$, c) $\neg S[j] \cap S[j+1]$ and $S[j-1] \cap S[j+1]$, and d) $S[j] \cap S[j+1]$ and $S[j-1] \cap S[j+1]$. We cover here only case d), as all the other cases use similar or easier reasoning. 
We consider score induced by sub-chains $S[j-1],S[j],S[j+1]$ and $S[j-1],S[j+1]$, respectively, assuming that the chains continue, so that the coverage induced by $S[j+1]$ will not yet be added to the total score. 
In case d) the score induced by the sub-chain $S[j-1],S[j],S[j+1]$ is $S[j].a-S[j-1].a+\min(S[j+1].a-S[j].a,S[j+1].c-S[j].c)=\min(S[j+1].a-S[j-1].a,S[j].a-S[j-1].a+S[j+1].c-S[j].c)$. The score induced by  the sub-chain $S[j-1],S[j+1]$ is $\min(S[j+1].a-S[j-1].a,S[j+1].c-S[j-1].c)$. Since $S[j].a-S[j-1].a\leq S[j].c-S[j-1].c$, the score induced by the sub-chain $S[j-1],S[j+1]$ is at least as high as the score induced by the sub-chain $S[j-1],S[j],S[j+1]$. 
\end{proof}

Shibuya and Kurochkin \cite{SK03} gave an $O(N \log N)$ time algorithm for Problem~\ref{prob:colinearchainingweak}, but their algorithm comes with no proof of correctness. Our goal in this paper is to complement the original proposal with the required derivation steps to see that one can indeed solve the problem correctly in $O(N \log N)$ time. Instead of proving directly the correctness of the original proposal, we derive a simplified version of the algorithm, whose correctness is easier to verify.

We derive this algorithm in three steps: First we consider one-sided overlaps of anchors. Then we modify this algorithm to handle two-sided overlaps of anchors, solving Problem~\ref{prob:colinearchaining}. Finally, we show that the use of strict precedence relation $I' \prec I$ can be relaxed to $I' \prec^w I$ in order to solve Problem~\ref{prob:colinearchainingweak}.

\section{Chaining algorithms}

Our goal is here to study the variations of chaining algorithms under the symmetric ordered coverage. We will give chaining algorithms under the symmetric ordered coverage and equal-match property taking $O(n\log n)$ time. In order to do this we will structure the recurrence relations that solve Problems 1 and 2 such that one can factor out dependencies between anchors into different cases that are handled by evaluation order of the recurrences, range search, and special features of the scoring function.   
Assume now that the anchor interval pairs are stored in an array $A[1..N]$ in arbitrary order. We fill a table $C[1\ltdots N]$ so that $C[j]$ gives the maximum symmetric ordered coverage of using any subset of pairs that precede $A[j]$, before the effect of the pair $A[j]$ is added to the score: Hence, $\max_j C^+[j]$, where $C^+[j]=\min(A[j].b-A[j].a+1,A[j].d-A[j].c+1)+ C[j]$, gives the total maximum symmetric ordered coverage.

After considering separately non-overlapping and overlapping cases (see Fig.~\ref{fig:precedence}), one observes that $C[j]$ can be computed by $\max(0,D[j],O[j])$, where
\begin{eqnarray*}
D[j]&=&
\begin{array}{c} 
\max\limits_{\begin{array}{c} j':\\ A[j']\prec A[j],\\ \neg A[j']\cap A[j]\end{array}} C[j']+\min\left\{\begin{array}{l} A[j'].b-A[j'].a+1,\\A[j'].d-A[j'].c+1\end{array}\right.\end{array} \text{ and}\\
O[j]&=&
\max\limits_{\begin{array}{c} j':\\ A[j']\prec A[j],\\ A[j']\cap A[j]\end{array}}  
C[j']+\min\left\{\begin{array}{l} \min(A[j].a,A[j'].b+1)-A[j'].a,\\\min(A[j].c,A[j'].d+1)-A[j'].c\end{array}
\right..
\end{eqnarray*}

These recurrences can be computed in $O(N^2)$ time: Sort $A$ by values $A[i].b$ to handle one dimension of the precedence relation. Then compute each $C[j]$ in this order by scanning previously computed values $C[j']$ and check precedence in the other dimension. Add the coverage values ($+\min$ part) depending on the overlap relation. Select the maximum among the options of $C[j']$ added with the coverage value.  

By assuming Equal Match Length property, we can simplify the recurrence of $C[j]=\max(0,D[j],O[j])$ with\\
\begin{eqnarray*}
D[j]&=& 
\max\limits_{\begin{array}{c} j':\\ A[j']\prec A[j],\\ \neg A[j']\cap A[j]\end{array}} C[j']+ A[j'].b-A[j'].a+1 \qquad\text{and}\\
O[j]&=&\max\limits_{\begin{array}{c} j':\\ A[j']\prec A[j],\\ A[j']\cap A[j]\end{array}} C[j']+\min\left\{\begin{array}{l} A[j].a-A[j'].a,\\A[j].c-A[j'].c\end{array}\right..
\end{eqnarray*}

\subsection{One-sided overlaps \label{sect:one-sided}}

We will now present an algorithm that works for \emph{one-sided overlaps} (see Fig.~\ref{fig:precedence}): We restrict the chains so that no two anchors in the solution overlap in the first dimension (that is, in $T$). This lets us modify the recurrence of $O[j]$ into the form\\
\begin{eqnarray*}
O[j]&=&A[j].c+\max\limits_{\begin{array}{c} j':\\ A[j']\prec A[j], A[j'].b<A[j].a,\\ A[j']\cap A[j]\end{array}} C[j']-A[j'].c\;\;.
\end{eqnarray*}

That is, we added the constraint on overlaps, removed the then obsolete $\min()$ and took out the value not affected by $\max()$. Now it is easy to see that the evaluation of the values can be done when visiting the starting points of the anchors in the first dimension, and the maximizations over range of values can be done using search trees, specified in the next lemma. We also specify a two-dimensional version of this structure, as we need it later.

\begin{lemma}
\label{lemma:searchtree}
The following three operations can be supported with a one-dimensional range search tree $\mathcal{T}$ in time $O(\log n)$, where $n$ is the number of search keys inserted to the tree.
\begin{itemize}
\item $\mathtt{Update}(k,\mathtt{val})$: Update $\mathtt{value}$ associated with $\mathtt{key}=k$ into $\mathtt{val}$.
\item $\mathtt{Upgrade}(k,\mathtt{val})$: Update $\mathtt{value}$ associated with $\mathtt{key}=k$ into $\max(\mathtt{val},\mathtt{value})$.
\item $\mathtt{RMaxQ}(c,d)$: Return maximum $\mathtt{value}$, where $c\leq \mathtt{key} \leq d$ (\emph{Range Maximum Query}).
\end{itemize}
Moreover, the search tree can be built in $O(n)$ time, given the $n$ pairs $(\mathtt{key},\mathtt{value})$ sorted by component $\mathtt{key}$. 
\end{lemma}

\begin{lemma}
\label{lemma:searchtree2d}
The following two operations can be supported with a two-dimensional range search tree $\mathcal{T}$ in time $O(\log^2 n)$, where $n$ is the number of search keys inserted to the tree.
\begin{itemize}
\item $\mathtt{Update}(p,s,\mathtt{val})$: Update $\mathtt{value}$ associated with $\mathtt{primary\;key}=p$ and $\mathtt{secondary\;key}=s$ into $\mathtt{val}$.
\item $\mathtt{Upgrade}(p,s,\mathtt{val})$: Update $\mathtt{value}$ associated with $\mathtt{primary\;key}=p$ and $\mathtt{secondary\;key}=s$ into $\max(\mathtt{val},\mathtt{value})$.
\item $\mathtt{RMaxQ}(a,b,c,d)$: Return maximum $\mathtt{value}$, where $a\leq \mathtt{primary\;key} \leq b$ and \\$c\leq \mathtt{secondary\;key} \leq d$ (\emph{2D Range Maximum Query}).
\end{itemize}
Moreover, the search tree can be built in $O(n\log n)$ time, given the $n$ triplets $(\mathtt{primary\;key},$ $\mathtt{secondary\;key},$ $\mathtt{value})$ sorted first by primary key and then by secondary key. 
\end{lemma}

These lemmas follow directly by maintaining maxima of values in each subtree for the corresponding standard range search structures \cite{thegeometrybook} that support listing all the (key, value) pairs in a range. Such constructions are often used in sparse dynamic programming \cite{Epp92,SK03,MNU05}.

\begin{algorithm}
\KwIn{A set of interval pairs $A[1..N]$ with all interval endpoints being distinct positive integers.}
\KwOut{Array $C^+[1..N]$ containing the symmetric ordered coverage values.}
Initialize one-dimensional search trees $\mathcal{T}^\mathtt{a}$ and $\mathcal{T}^\mathtt{b}$ with keys $A[j].d$, $1 \leq j \leq N$, and with key $0$, all keys associated with values $-\infty$\;
$\mathcal{T}^\mathtt{a}.\mathtt{Upgrade}(0,0)$\; 
$E=\{(A[j].a,j) \mid 1\leq j \leq N\} \cup \{(A[j].b,j) \mid 1\leq j \leq N\}$\; 
$E.sort()$\; 
\For{$i \gets 1$ to $2N$}{
  $j=E[i][2]$\;
  $I = A[j]$\;
  \If{$I.a==E[i][1]$}{
      $C^\mathtt{a}[j] = \mathcal{T}^\mathtt{a}.\mathtt{RMaxQ}(0,I.c-1)$\;
      $C^\mathtt{b}[j] = I.c+ \mathcal{T}^\mathtt{b}.\mathtt{RMaxQ}(I.c,I.d)$\;
      $C[j] = \max(C^\mathtt{a}[j],C^\mathtt{b}[j])$\;
      $C^+[j] = C[j] + I.b-I.a+1$\;
  }
  \Else{
    $\mathcal{T}^\mathtt{a}.\mathtt{Upgrade}(I.d,C^+[j])$\;
    $\mathcal{T}^\mathtt{b}.\mathtt{Upgrade}(I.d,C[j]-I.c)$\;
  }  
}
\Return{$C^+[1..N]$}\;
\caption{\label{algo:colinearchainingOneSidedOverlaps}
Chaining allowing one-sided overlaps.}
\end{algorithm}

We obtain Algorithm~\ref{algo:colinearchainingOneSidedOverlaps} to handle the one-sided overlaps case, where we have replaced arrays $O$ and $D$ with $C^\mathtt{a}$ and $C^\mathtt{b}$, respectively, to reflect the cases shown in Fig.~\ref{fig:precedence}.

The pseudocode of Algorithm~\ref{algo:colinearchainingOneSidedOverlaps} assumes interval endpoints to be distinct. This assumption is only used for the ease of presentation. It can be relaxed by the standard method used in computational geometry: Replace each endpoint $x$ by a pair $(x,j)=E[i]$ where $A[j]$ identifies the anchor in question. These pairs $E[i]=(x,j)$ are distinct, and can be used as the keys of the search trees (in place of just $x$). Range queries can be implemented to ignore the secondary key $j$.   

\begin{lemma}
\label{lemma:colinearchaining}
Problem~\ref{prob:colinearchaining} on $N$ input pairs restricted to solutions that contain only one-sided overlaps can be solved in $O(N \log N)$ time, assuming the input satisfies Equal Match Length property.
\end{lemma}
\begin{proof}
The evaluation order of Algorithm~\ref{algo:colinearchainingOneSidedOverlaps} guarantees that when computing the values $C^\mathtt{a}[j]$ and $C^\mathtt{b}[j]$, the data structures contain only anchors that precede the current anchor and do not overlap it in the first dimension. The range query on $\mathcal{T}^\mathtt{a}$ guarantees that we also consider only those anchors that precede and do not overlap in the second dimension for the computation of $C^\mathtt{a}[j]$. The range query on $\mathcal{T}^\mathtt{b}$ guarantees that we also consider only those anchors that overlap in the second dimension for the computation of $C^\mathtt{b}[j]$, but this is not enough to guarantee predecessor-relation to hold. That is, there can be an anchor $I'$ stored in $\mathcal{T}^\mathtt{b}$ with $I'.b<I.a$ and $I.c\leq I'.c\leq I'.d< I.d$ and thus the evaluation order and range query fail to guarantee $I'.c< I.c$ to make $I'\prec I$ (recall the definition).
We need to show that if such $I'$ is in an optimal chain to $I$, there is always another optimal chain to $I$ not including $I'$. Consider the last anchor $A[j'']=I''$ in an optimal chain to $A[j']=I'$ that overlaps and precedes $I$. Then we know that $C[j]\geq C[j'']+I.c-I''.c$ and $C[j']\leq C[j'']+I'.c-I''.c$, so a chain where $I''$ directly precedes $I$ does not decrease the score. If such $I''$ does not exist but an optimal chain to $I$ includes $I'$, we have that $\max(C^\mathtt{a}[j],C^\mathtt{b}[j])=C^\mathtt{a}[j]$, as all anchors in an optimal chain to $I'$, excluding $I'$, are stored in $\mathcal{T}^\mathtt{a}$, and including $I'$ can only decrease the score as $I.c-I'.c\leq 0$.
\end{proof}

\subsection{Two-sided overlaps}

The trick by Shibuya and Kurochkin \cite{SK03} to handle two-sided overlaps is to separate them to two cases (see Fig.~\ref{fig:precedence}): (c) overlaps in the first dimension are at least as long as in the second dimension and (d) overlaps are longer in the second dimension. Since our algorithm so far considers all anchors that do not overlap in the first dimension, it will be enough to consider how to enhance the algorithm to handle anchors that do overlap in the first dimension. 

Consider case (c). That is, for any two pairs of anchors $I',I$, $I'\prec I$, it holds $I'.a< I.a \leq I'.b < I.b$, $I'.c< I.c$, $I'.d < I.d$ and $I'.d-I.c \leq I'.b-I.a$. The latter inequality can be written as $I.c-I.a \geq I'.c-I'.a$ (due to Equal Match Length property). Also, if $I'$ precedes $I$ in an optimal chain to $I$, the score calculated up to $I'$ will increase by inclusion of $I$ by $\min(I.a-I'.a,I.c-I'.c)=I.a-I'.a$ (due to Equal Match Length property). This means that once we first stop at anchor $I=A[j]$ in our algorithm, if we have inserted to a search tree $T^{\mathtt{c}}$ all anchors $A[j']$ that overlap $I$ in the first dimension, using keys $A[j'].c-A[j'].a$ and values $C[j']-A[j'].a$, we can query $T^{\mathtt{c}}.\mathtt{RMaxQ}(-\infty,A[j].c-A[j].a)$ and add $A[j].a$ to obtain the correct score for this case. However, in the order we evaluate the anchors we can only guarantee $A[j'].a < A[j].a \leq A[j'].b$ and thus $A[j'].c<A[j].c$ (property of case (c)), but not $A[j'].b < A[j].b$ or $A[j'].d < A[j].d$. To solve this, we add another dimension to the search tree, so we can add constraint $A[j'].b < A[j].b$ to the query, which also covers the remaining constraint $A[j'].d < A[j].d$ (property of case (c)).

Case (d) is almost symmetric to case (c): For any two pairs of anchors $I',I$, $I'\prec I$, it holds $I'.a< I.a \leq I'.b < I.b$, $I'.c< I.c$, $I'.d < I.d$ and $I'.d-I.c > I'.b-I.a$. The latter inequality can be written as $I'.c-I'.a>I.c-I.a$. Also, if $I'$ precedes $I$ in an optimal chain to $I$, the score calculated up to $I'$ will increase by inclusion of $I$ by $\min(I.a-I'.a,I.c-I'.c)=I.c-I'.c$ (due to Equal Match Length property). This means that once we first stop at anchor $I=A[j]$ in our algorithm, if we have inserted to a search tree $T^{\mathtt{d}}$ all anchors $A[j']$ that overlap $I$ in the first dimension, using keys $A[j'].c-A[j'].a$ and values $C[j']-A[j'].c$, we can query $T^{\mathtt{c}}.\mathtt{RMaxQ}(A[j].c-A[j].a+1,\infty)$ and add $A[j].c$ to obtain the correct score for this case. As before, we need to add another dimension to the search tree to handle constraint $A[j'].d < A[j].d$, which also covers constraint $A[j'].b < A[j].b$ (property of case (d)). We are left with constraints  $A[j'].a < A[j].a \leq A[j'].b$ and $A[j'].c<A[j].c$, where the first ones follow from the evaluation order, but now the latter is not automatically guaranteed to hold: Using arguments analogous to the proof of Lemma~\ref{lemma:colinearchaining}, we show that such nested case cannot change the optimal solution. 

The resulting enhancement to handle two-sided overlaps is given as Algorithm~\ref{algo:colinearchainingTwoSidedOverlaps}.

\begin{algorithm}
\KwIn{A set of interval pairs $A[1..N]$ with all interval endpoints being distinct positive integers.}
\KwOut{Array $C^+[1..N]$ containing the symmetric ordered coverage values.}
Initialize one-dimensional search trees $\mathcal{T}^\mathtt{a}$ and $\mathcal{T}^\mathtt{b}$ with keys $A[j].d$, $1 \leq j \leq N$, and with key $0$, all keys associated with values $-\infty$\;
Initialize two-dimensional search trees $\mathcal{T}^\mathtt{c}$ and $\mathcal{T}^\mathtt{d}$ with keys $(A[j].c-A[j].a,A[j].b)$ and $(A[j].c-A[j].a,A[j].d)$, respectively, for $1 \leq j \leq N$, associated with values $-\infty$\;
$\mathcal{T}^\mathtt{a}.\mathtt{Upgrade}(0,0)$\; 
$E=\{(A[j].a,j) \mid 1\leq j \leq N\} \cup \{(A[j].b,j) \mid 1\leq j \leq N\}$\; 
$E.sort()$\; 
\For{$i \gets 1$ to $2N$}{
  $j=E[i][2]$\;
  $I = A[j]$\;
  \If{$I.a==E[i][1]$}{
      $C^\mathtt{a}[j] = \mathcal{T}^\mathtt{a}.\mathtt{RMaxQ}(0,I.c-1)$\;
      $C^\mathtt{b}[j] = I.c+ \mathcal{T}^\mathtt{b}.\mathtt{RMaxQ}(I.c,I.d)$\;
      $C^\mathtt{c}[j] = I.a+ \mathcal{T}^\mathtt{c}.\mathtt{RMaxQ}(-\infty,I.c-I.a,0,I.b)$\;
      $C^\mathtt{d}[j] = I.c+ \mathcal{T}^\mathtt{d}.\mathtt{RMaxQ}(I.c-I.a+1,\infty,0,I.d)$\;      
      $C[j] = \max(C^\mathtt{a},C^\mathtt{b},C^\mathtt{c},C^\mathtt{d})$\;
      $C^+[j] = C[j]+I.b-I.a+1$\;
      $\mathcal{T}^\mathtt{c}.\mathtt{Upgrade}(I.c-I.a,I.b,C[j]-I.a)$\;
      $\mathcal{T}^\mathtt{d}.\mathtt{Upgrade}(I.c-I.a,I.b,C[j]-I.c)$\;
  }
  \Else{
    $\mathcal{T}^\mathtt{a}.\mathtt{Upgrade}(I.d,C^+[j])$\;
    $\mathcal{T}^\mathtt{b}.\mathtt{Upgrade}(I.d,C[j]-I.c)$\;
    $\mathcal{T}^\mathtt{c}.\mathtt{Update}(I.c-I.a,I.b,-\infty)$\;
    $\mathcal{T}^\mathtt{d}.\mathtt{Update}(I.c-I.a,I.b,-\infty)$\;
  }
}
\Return{$C^+[1..N]$}\;
\caption{\label{algo:colinearchainingTwoSidedOverlaps}
Chaining with two-sided overlaps.}
\end{algorithm}

The pseudocode of Algorithm~\ref{algo:colinearchainingTwoSidedOverlaps} assumes interval endpoints to be distinct, but this can be relaxed as in the proof of Lemma~\ref{lemma:colinearchaining}. Using the data structure from Lemmas~\ref{lemma:searchtree}~and~\ref{lemma:searchtree2d} we obtain the following result.

\begin{theorem}
\label{theorem:colinearchainingTwoSidedOverlaps}
Problem~\ref{prob:colinearchaining} on $N$ input pairs can be solved in  $O(N \log^2 N)$ time (by Algorithm~\ref{algo:colinearchainingTwoSidedOverlaps}), assuming the input satisfies Equal Match Length property.
\end{theorem}
\begin{proof}
As discussed earlier, it is sufficient to consider anchors $A[j']$ and $A[j]$ that satisfy the precedence and overlap relations except for $A[j'].c < A[j].c$ not holding, as all other constraints are properly covered by the combination of evaluation order and the queries. Such invalid anchors $A[j']$ can affect the query results from data structures $\mathcal{T}^\mathtt{b}$ and $\mathcal{T}^\mathtt{d}$ when computing the score for $A[j]$. Consider that an optimal chain to $A[j]$ has $A[j']$ as the previous anchor and thus $C[j]=C[j']+A[j].c-A[j'].c$. Consider the last anchor $A[j'']$ in an optimal chain to $A[j']$ that precedes and overlaps $A[j]$. Assume the overlap is larger in the first dimension (the other case is considered already in the proof of Lemma~\ref{lemma:colinearchaining}). Then $C[j']= C[j'']+A[j].a-A[j''].a$ as $A[j']$ must overlap $A[j]$ in the first dimension, $A[j'']$ must directly precede $A[j']$ for it being the last with this property, and the overlap between $A[j'']$ and $A[j']$ is larger in the first dimension due to transitivity. As $A[j].a-A[j'].a\geq A[j].c-A[j'].c$, the direct use of $A[j'']$ before $A[j]$ gives $C[j]\geq C[j'']+A[j].a-A[j''].a=C[j'']+A[j].a-A[j'].a+A[j'].a-A[j''].a\geq C[j']+A[j].c-A[j'].c$. That is, $A[j']$ can be omitted from the optimal path.  
\end{proof}

\subsection{Overlaps with weak precedence}

Let us now proceed to improve the running time of Algorithm~\ref{algo:colinearchainingTwoSidedOverlaps} to $O(N \log N)$ by considering chains under the weak precedence relation (Problem~\ref{prob:colinearchainingweak}). For this, we drop the second dimension of the data structures $\mathcal{T}^\mathtt{c}$ and $\mathcal{T}^\mathtt{d}$, that were added to guarantee strict precedence. However, this is not sufficient for proving correctness as we used these constraints to indirectly guarantee precedence of start positions of anchors as well. Case (c) causes no problems, as the evaluation order guarantees that $\mathcal{T}^\mathtt{c}$ contains anchors $A[j']$ with $A[j'].a<A[j].a$ and the query restricts to cases $A[j'].c<A[j].c$. However, in case (d) the solution returned can have $A[j].c\leq A[j'].c$. We will consider this in the proof of the next theorem.

\begin{theorem}
\label{theorem:colinearchainingWeak}
Problem~\ref{prob:colinearchainingweak} on $N$ input pairs can be solved in  $O(N \log N)$ time (by Algorithm~\ref{algo:colinearchainingTwoSidedOverlaps} with the operations on the second dimension of search trees $\mathcal{T}^\mathtt{c}$ and $\mathcal{T}^\mathtt{d}$ omitted), assuming the input satisfies Equal Match Length property.
\end{theorem}
\begin{proof}
As discussed, it is sufficient to show that queries from $\mathcal{T}^\mathtt{d}$ correspond to proper solutions. For contradiction, assume that $C[j]=C^\mathtt{d}[j]$,  $C^\mathtt{d}[j]>\max(C^\mathtt{a}[j],C^\mathtt{b}[j],C^\mathtt{c}[j])$, and $C[j]=C[j']+A[j].c-A[j'].c$ only for $A[j']$s for which $A[j].c\leq A[j'].c$. 
Such solution is not proper (weak precedence not holding), so we need to show that there is an equivalently good proper solution. 

First, if it also holds $A[j'].d\leq A[j].d$, we have the nested case handled already in the proof of Theorem~\ref{theorem:colinearchainingTwoSidedOverlaps}. We continue with the case case where $A[j].c\leq A[j'].c$ and $A[j].d<A[j'].d$ hold. This setting is illustrated in Fig.~\ref{fig:casebproof}.

Consider an anchor $A[j'']$ in an optimal chain to $A[j']$ that overlaps position $A[j].c$ in the second dimension. It holds $C[j']\leq C[j'']+A[j'].c-A[j''].c$, since any chain from $A[j'']$ to $A[j']$ can cover at most $A[j'].c-A[j''].c$ positions. But then there is an optimal chain to $A[j]$ avoiding $A[j']$ with score $C[j]\geq C[j'']+A[j].c-A[j''].c=C[j'].c-A[j'].c+A[j''].c+A[j].c-A[j''].c=C[j']+A[j].c-A[j'].c$, which is a contradiction. 

We are left with the case that there is no such $A[j'']$ in the optimal chain to $A[j']$ that overlaps position $A[j].c$ in the second dimension. Let then $A[j'']$ be the last anchor in an optimal chain to $A[j']$ that does not overlap $A[j].c$. We have three cases to consider: a) $\neg A[j''] \cap A[j']$, b) $A[j''] \cap A[j']$, and c) no such $A[j'']$ exist.

In case a) $C[j']< C^+[j'']+A[j'].c-A[j].c$, assuming that only $A[j].c$ is left uncovered between anchors $A[j'']$ and $A[j']$. Using this we can write $C[j]=C[j']+A[j].c-A[j'].c< C^+[j'']+A[j'].c-A[j].c +A[j].c-A[j'].c= C^+[j'']$. Since it holds $A[j''] \prec^w A[j]$ and $\neg A[j''] \cap A[j]$ we have that $C[j]\geq C^+[j'']$ using directly $A[j'']$ avoiding $A[j']$. This is contradiction.

In case b) $A[j'']$ can only overlap $A[j']$ in the first dimension. Then $C[j']=C[j'']+A[j'].a-A[j''].a$. If $A[j'']$ also overlaps $A[j]$, it can do so only in the first dimension. Then $C[j]\geq C[j'']+A[j].a-A[j''].a\geq C[j'']+A[j'].a-A[j'].a$. That is, $A[j']$ can be avoided in an optimal chain to $A[j]$ by using $A[j'']$ instead. On the other hand, if $A[j'']$ does not overlap $A[j]$, we have $C[j]\geq C^+[j'']\geq C[j'']+A[j'].a-A[j''].a=C[j']$, in which case we also get a contradiction.

In case c) $C[j']=0$ and thus $C[j]=C[j']+A[j].c-A[j'].c\leq 0$, which contradicts our assumption $C^\mathtt{d}[j]>C^\mathtt{a}[j]=0$. \end{proof}

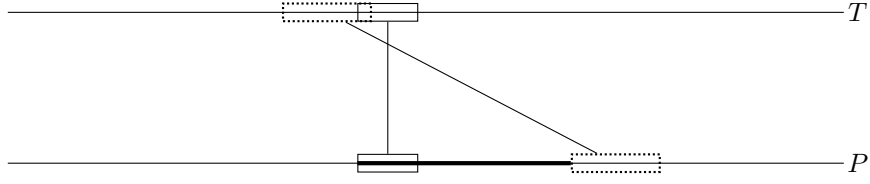
\begin{figure}
        \centering
        \begin{tikzpicture}
        \draw (0,2) to (11,2) node {\verb+  +$T$};
        \draw (0,0) to (11,0) node {\verb+  +$P$};
        \draw[densely dotted,thick] (4.2,2) node[rectangle,draw](Iptop) {\verb+      +};
        \draw (5,2) node[rectangle,draw](Itop) {\verb+    +};
        \draw[densely dotted,thick] (8,0) node[rectangle,draw](Ipbottom) {\verb+      +};
        \draw (5,0) node[rectangle,draw](Ibottom) {\verb+    +};
        \draw[-] (Iptop) to (Ipbottom);
        \draw[-] (Itop) to (Ibottom);
        \draw[ultra thick] (4.6,0) to (Ipbottom);
        \end{tikzpicture}        
    \caption{Dotted and solid rectangles denote $A[j']$ and $A[j]$, respectively. Here the (weak) precedence is not holding as the interval of $A[j']$ in the second dimension succeeds the corresponding interval of $A[j]$. The thick line segment represents maximum coverage a chain ending at $A[j']$ (not including $A[j']$) can achieve after starting from $A[j]$. The algorithm subtracts this thick line segment length from the score, so that there is at least as good chain to $A[j]$ that avoids using $A[j']$. \label{fig:casebproof}}
\end{figure}

\section{Connection to LCS\label{sect:lcs}}

String $C$ is a \emph{Longest Common Subsequence} (LCS) of strings $T$ and $P$ if it is a longest string that can be obtained by deleting $0$ or more characters from both $T$ and from $P$. Such $C[1..\ell]$ can be written as $T':=T[i_1]T[i_2]\cdots T[i_\ell]$ and as $P':=P[j_1]P[j_2]\cdots P[j_\ell]$, where $1\leq i_1 < i_2 < \cdots < i_\ell \leq |T|$ and $1\leq j_1 < j_2 < \cdots < j_\ell \leq |P|$. Consider the set of anchors $A$ being exact matches between $T$ and $P$. We say that $C$ is an \emph{anchor-restricted LCS} if
it can be written as $T'$ and as $P'$ defined above such that for each $(i_k,j_k)$ there is an anchor $([a\ltdots b],[c\ltdots d])$ in $A$ with $a+x=i_k$ and $c+x=j_k$ for some $x$, $0\leq x \leq b-a=d-c$. Informally, such $C$ is a longest string with all characters appearing in increasing order in $T$ and $P$ where each such occurrence of a character is \emph{supported} by at least one anchor. We show that an \emph{anchor-restricted LCS} can be found by solving the problem of chaining under the weak precedence:

\begin{theorem}
Assume the anchors $A$ are exact matches between input strings $T$ and $P$. The score of a chain $S$ such that $\mathtt{coverage}(S)=\max_{1 \leq i \leq N} \mathtt{coverage}(S^i)$ of Problem~\ref{prob:colinearchainingweak} equals the length of an \emph{anchor-restricted LCS} of $T$ and $P$.
\label{thm:lcs}
\end{theorem}
\begin{proof}
Due to Lemma~\ref{lemma:weakisenough}, we can assume $S$ is a chain under the strict precedence order. Each anchor in $S$ contributes to the score by the minimum length of its intervals after the overlaps with the previous anchor intervals have been cut out. This minimum length equals the number of characters that can be included to the common subsequence. That is, we can extract an anchor-restricted subsequence of $T$ and $P$ of length $\mathtt{coverage}(S)$ from the solution. We need to show that such subsequence is the longest among anchor-restricted subsequences. Assume, for contradiction, that there is an anchor-restricted LCS $C[1..\ell]$ longer than $\mathtt{coverage}(S)$. Consider the chain of $\ell$ anchors formed by taking for each $C[k]$ an anchor containing match $T[i_k]=P[j_k]$. Assign a score $1$ to each anchor included in the chain. Let us modify this chain into a chain where weak precedence holds such that the total score (number of matches induced by the solution) remains the same. First, we merge from left to right all runs of identical anchors; score of an anchor is then the length of the run in the original chain. Then we consider anchors from left to right. Consider the first pair of anchors $I',I$ in the current chain for which $I.a\leq I'.a$ (case $I.c\leq I'.c$ is symmetric). Let the score of $I'$ be $x$. By construction, we know that the left-most position possible for the first match of $I$ included in $C$ is $I'.a+x$. Therefore, we can remove $I'$ from the chain and include $x$ matches from $I$. The total score does no decrease by this change. This process can be repeated until the weak precedence relation holds up to $I$, and then continued similarly to the end of the chain, yielding a contradiction. 
\end{proof}

\bibliography{document}

\end{document}